\documentclass[11pt]{article}
\usepackage{natbib,amssymb,amsmath,amsthm,fullpage,hyperref,graphicx}

\newcommand{\Exp}[1]{{\text{E}}[ \ensuremath{ #1 } ]  }
\newcommand{\Var}[1]{{\text{Var}}[ \ensuremath{ #1 } ]  }

\newcommand{\bl}[1]{{\mathbf #1}}
\newcommand{\bs}[1]{{\boldsymbol #1}}
\newcommand{\tr}{\text{tr}}

\newtheorem{lem}{Lemma}

\newtheorem{thm}{Theorem}
\newtheorem{cor}{Corollary}

\title{Exact adaptive confidence intervals for linear regression coefficients}
\author{Peter D. Hoff$^{1}$ and  Chaoyu Yu$^{2}$  \\ 
 $^1$Department of Statistical Science, Duke University \\
 $^2$Department of Biostatistics, University of Washington-Seattle}

\begin{document} 

\maketitle

\begin{abstract}
We propose an adaptive confidence interval procedure (CIP)
for the coefficients in the normal linear regression model. 
This procedure has a frequentist coverage rate that is constant as a
function of the model parameters, yet provides smaller 
intervals 
than
the usual
interval procedure, 
on average across regression coefficients. 
The proposed procedure is obtained by defining a class of CIPs 
that all have exact  $1-\alpha$ frequentist coverage, and then 
selecting from this class the procedure that minimizes 
a prior expected interval width. Such a procedure 
may be described as ``frequentist, assisted by Bayes'' 
or FAB.  We describe an adaptive approach for estimating 
the prior distribution from the data so that 
exact non-asymptotic $1-\alpha$ coverage is maintained. 
Additionally, in a ``$p$ growing with $n$'' asymptotic scenario, 
this adaptive FAB procedure is asymptotically Bayes-optimal
among $1-\alpha$ frequentist CIPs. 

\end{abstract}  

\section{Introduction}
Linear regression analyses routinely include 
point estimates and confidence intervals for the 
regression coefficients $\bs\beta=(\beta_1,\ldots, \beta_p)$ of the 
linear model $\bl y \sim N_n(\bl X\bs\beta , \sigma^2 \bl I)$. 
The most widely-used confidence interval procedure (CIP)
for an element $\beta_j$ of $\bs\beta$ is perhaps 
the usual $t$-interval centered around
the ordinary least-squares (OLS) estimate $\hat\beta_j$. 
This interval is uniformly most accurate among CIPs
that are derived from inversion of unbiased tests, and so it is called the 
uniformly most accurate unbiased (UMAU) CIP. 

In this article we consider alternatives to the UMAU procedure 
that have \emph{constant coverage}, that is, interval procedures 
$C_j(\bl y)$ satisfying 
\begin{equation}
\Pr( \beta_j \in C_j(\bl y) |\bs\beta, \sigma )  = 1-\alpha, 
 \ \forall  \, (\bs\beta,\sigma) \in \mathbb R^p\times \mathbb R^+. 
 \label{eqn:ccoverage} 
\end{equation} 
This property is what we normally think of as the 
usual frequentist definition of 
$1-\alpha$ coverage -
the random interval $C_j(\bl y)$ covers the true value 
$\beta_j$ with probability $1-\alpha$, no matter 
what $\bs\beta$ and $\sigma$ are. 
We introduce the term  ``constant coverage'' to distinguish such intervals  
%having exact and constant frequentist coverage from 
from other intervals whose coverage is bounded below by $1-\alpha$ 
but varies with $(\bs \beta,\sigma^2)$, or 
so-called ``frequentist 
intervals'' whose coverage rate is only constant as a function 
of the parameters asymptotically. For example, the usual 
score interval 
for a coefficient in a logistic regression model 
has an actual finite-sample 
coverage rate that depends on the values of the parameters.

The UMAU interval procedure of course has constant coverage, and 
it also has an expected width that is constant for all values 
of $\bs\beta$. 
However, in many cases we have prior information that 
many of the elements of $\bs\beta$ may be close to a
particular value, such as zero. In this case, we might 
prefer  a CIP for $\beta_j$  that has a smaller expected width
for ``likely'' values of $\beta_j$  in exchange for 
having wider intervals for values of $\beta_j$ that are less likely. 
Specifically, if our prior information could be quantified in terms of 
a prior distribution with density $\pi(\bs\beta)$, 
then arguably we would 
be interested in an interval procedure 
that 
minimizes the prior expected width 
\begin{equation*}
  \Exp{  |C_j(\bl y)|  } = \int \int |C_j(\bl y)|  \, 
                           p(\bl y | \bs\beta,\sigma) \, d\bl y \, 
    \pi(\bs\beta)  d\bs \beta % \label{eqn:bwidth}
\end{equation*}
among all CIPs that satisfy  the constant coverage property 
(\ref{eqn:ccoverage}).  Such a procedure would still be 
``frequentist'' in that it would have $1-\alpha$ constant coverage, 
but it would also be Bayes-optimal among frequentist procedures. 
We refer to such a statistical 
procedure as ``frequentist, assisted by Bayes'' 
or FAB.  

In practice, an appropriate prior distribution may not be known 
in advance. In this article we present a method for adaptively 
estimating a normal prior distribution for $\bs\beta$ 
from the data $\bl y$, and then using 
this estimated prior distribution to construct an approximately 
Bayes-optimal CIP for each regression 
coefficient $\beta_j, j=1,\ldots, p$.  The CIP we propose 
satisfies the constant coverage condition (\ref{eqn:ccoverage})
exactly and non-asymptotically, but it is also 
Bayes-optimal asymptotically as $p$ and $n$ increase to infinity. 
Our proposed adaptive CIP builds on the work of 
\citet{pratt_1963}, who obtained a Bayes-optimal frequentist 
confidence interval for the mean of a normal population with 
a known variance. In the next section we review Pratt's 
FAB interval, and discuss an extension developed in \citet{yu_hoff_2016}
to accommodate an unknown variance. 
In Section 3 we 
further extend these ideas to the case of interval estimation 
for a linear regression coefficient, 
and show how we may  adaptively estimate  
a normal prior distribution for the elements of $\bs\beta$. 
The resulting 
adaptive FAB confidence interval we propose  maintains
exact, non-asymptotic constant coverage. 
Additionally,  since the accuracy of our adaptive estimate improves as 
$n$ and $p$ increase, our adaptive FAB procedure 
is Bayes-optimal under this type of asymptotic regime. 
Section 4 includes several numerical examples illustrating the 
use of the adaptive FAB procedure, including analyses of two  
datasets and a small simulation study. A discussion follows in Section 5. 

Several other authors have studied alternatives 
to UMAU intervals for regression parameters. 
\citet{ogorman_2001} developed a CIP based on a permutation test
that adapts to non-normal error distributions, as 
opposed to adapting to 
small or sparse values of the 
regression coefficients. 
\citet{kabaila_dilshani_2014} developed
an improved CIP based on an adaptive estimate of the 
variance $\sigma^2$. Their procedure depends on a user-specified 
spline function for which the  constant coverage property must 
be checked numerically. In contrast, 
our proposed CIP is obtained by adaptively 
selecting from a class of constant-coverage intervals 
based on easy to obtain 
estimates of a few parameters. For our procedure, constant 
coverage follows by construction and does not need to be checked numerically. 
\citet{lee_sun_sun_taylor_2016} developed a procedure
that has exact conditional coverage, given a model selection event and 
knowledge of $\sigma^2$. 
However, for cases where $\sigma^2$ is unknown, their 
suggested modification uses a plug-in estimate of $\sigma^2$ and 
achieves exact coverage only asymptotically.  
Other authors (\citet{buhlmann_2013}, \citet{degeer_etal_2014}, \citet{zhang_zhang_2014}) 
have considered confidence interval construction for sparse, 
high-dimensional regression, including the case that $p>n$. These 
approaches generally work by de-biasing sparse estimators of 
the regression coefficients. However, the coverage rates 
of these methods are asymptotic, and typically depend on conditions 
on the design matrix $\bl X$ and the degree of sparsity of 
$\bs \beta$. For example, in Section 4 we show that the finite-sample  
coverage of one such procedure can be very good in a sparse setting, 
but extremely poor if $\bs\beta$ is not sparse.

\section{Review of FAB intervals }
Suppose  $\hat\theta$  is normally distributed with unknown mean $\theta$ and 
known variance $\sigma^2$. 
Then for any choice of $s\in[0,1]$,
\begin{align*}\
\Pr(  z_{\alpha(1-s)} < (\theta-\hat\theta )/\sigma <
       z_{1-\alpha s} | \theta)  &=   (1-\alpha s) - \alpha (1-s) \\
 &  = 1- \alpha,
\end{align*}
where $z_p$ denotes the $p$th quantile of the standard normal distribution. 
As described in \citet{yu_hoff_2016}, this implies that 
for any function $s:\mathbb R \rightarrow [0,1]$ 
the set-valued function
\begin{equation}
C_s(\hat \theta) = 
\left \{ \theta : 
   \hat\theta + \sigma z_{\alpha(1-s(\theta))} < \theta < \hat\theta + \sigma z_{1-\alpha s(\theta)}  \right  \}
\label{eqn:gen_ci}
\end{equation}
is a $1-\alpha$ confidence procedure, satisfying
$\Pr( \theta \in C_s(\hat\theta) | \theta )=1-\alpha$.  
We refer to such a function $s(\theta)$ as a \emph{spending function}, 
as it corresponds to regions of the parameter space upon which 
type I error is ``spent.''

The usual procedure is $C_{1/2}(\hat \theta)$, obtained from
the constant  spending function $s(\theta)=1/2$. 
While $C_{1/2}$ is 
the uniformly most accurate unbiased 
(UMAU) confidence interval procedure (CIP), 
the lack of a uniformly most powerful 
test of $H_\theta:\Exp{\hat\theta}=\theta$ 
versus $K_\theta:\Exp{\hat\theta}\neq\theta$  means
there are  confidence procedures corresponding
to collections of biased level-$\alpha$ tests that
have smaller expected widths
than the UMAU procedure
for some regions of the parameter space.
If prior information is available that
$\theta$ is likely to be near some value $\mu$,
then we may be willing to incur
wider intervals  for $\theta$-values far from $\mu$
in exchange for smaller intervals near $\mu$.
With this in mind, \citet{pratt_1963}
developed a Bayes-optimal $1-\alpha$ CIP 
that minimizes the ``Bayes width'' or expected 
interval width averaged over values of both
$\hat\theta$ and $\theta$, where the latter averaging
is done with respect to a $N(\mu,\tau^2)$  prior distribution for
$\theta$. The resulting CIP has $1-\alpha$ frequentist
coverage for each value of $\theta$, but has lower
expected width for values of $\theta$ near the prior
mean (and wider expected widths elsewhere). 
We describe this interval as being 
``frequentist assisted by Bayes''
or FAB.
As shown in \citet{yu_hoff_2016}, 
the  spending function corresponding to Pratt's FAB 
confidence interval is characterized as follows: If $\tau^2>0$, then 
\begin{align} 
\label{eqn:sfabz}
 s(\theta) & = g^{-1} ( 2\sigma(\theta-\mu)/\tau^2 )  \\
 g(s)  & =  \Phi^{-1}(\alpha s) - \Phi^{-1}(\alpha(1-s) ). \nonumber 
\end{align}
If $\tau^2=0$, then $s(\theta)=1$ if $\theta>\mu$ and 
 $s(\theta)<0$ if $\theta<  \mu$. 
The value of $s(\mu)\in[0,1]$ does not affect the 
width of the confidence interval, but can affect whether or not 
$\mu$ is included in the interval or not (as an endpoint). 
We suggest taking $s(\mu)$ to be 1/2 when $\tau^2=0$, as it is 
in the case that $\tau^2>0$.

Now consider 
confidence interval construction for $\theta$ in the more typical case that 
$\sigma^2$ is unknown. Suppose we will observe independent statistics  
$\hat \theta$ and $\hat \sigma^2$, where 
$\hat\theta \sim N(\theta,\sigma^2)$ and  
$q \hat \sigma^2/\sigma^2 \sim   \chi^2_q$.   
Letting $t_p$ be the $p$th quantile of the $t$-distribution with
$q$ degrees of freedom,
any spending function $s:\mathbb R \rightarrow [0,1]$ 
defines a class of acceptance regions 
\[
  A_s(\theta) =\left  \{ (\hat \theta ,\hat\sigma) : 
 \hat\theta + \hat \sigma t_{\alpha(1-s(\theta))} < \theta < \hat\theta + \hat\sigma  t_{1-\alpha s(\theta)}  \right \},        
\]
so that for each $\theta$, $A_s(\theta)$ is the acceptance region 
of a level-$\alpha$ test. Inversion of this class of tests 
yields a CIP with exact $1-\alpha$ constant coverage, 
\begin{equation}
C_s(\hat \theta,\hat\sigma) = 
\left \{ \theta : 
   \hat\theta + \hat \sigma t_{\alpha(1-s(\theta))} < \theta < \hat\theta + \hat\sigma  t_{1-\alpha s(\theta)}  \right \}. 
\label{eqn:fab_ci}
\end{equation}
%where $t_p$ is the $p$th quantile of the $t$-distribution with
%$q$ degrees of freedom. 
Important properties of
this CIP include the following:
\begin{thm} 
\label{prp:tprop}
If $\hat\theta \sim N(\theta, \sigma^2)$ and
$ q \hat \sigma^2/\sigma^2 \sim \chi^2_q $ are independent, then
\begin{enumerate}
\item $\Pr( \theta\in C_s(\hat \theta,\hat\sigma) | \theta,\sigma )=1-\alpha \, \forall \, (\theta,\sigma )\in \mathbb R\times \mathbb R^+$,
so the procedure has $1-\alpha$ constant coverage;
\item If $s(\theta)$ is nondecreasing then $C_s(\hat\theta,\hat \sigma)$ is
an interval with probability 1;
\item If $s(\theta)$ is not nondecreasing then $C_s(\hat\theta,\hat \sigma)$ is
an interval with probability less than 1.
\end{enumerate}
\end{thm}
Items 1 and 2 were shown in \citet{yu_hoff_2016}, and a proof of 
item 3 is in the appendix. 
This result says that
every  spending function corresponds to a $1-\alpha$ frequentist
confidence procedure, and every \emph{nondecreasing}  spending function
corresponds to a $1-\alpha$ confidence \emph{interval} procedure.  
 \citet{yu_hoff_2016} showed that the
spending function (\ref{eqn:sfabz}) that corresponds to Pratt's 
$z$-interval is strictly increasing. 
If such a nondecreasing spending function $s$ is used, then  the
lower and upper endpoints of the interval, $ \underline{\theta}$ and $\overline{\theta}$,  are obtained by solving the equations
\begin{equation} 
F \left (\tfrac{\underline{\theta} -\hat \theta }{\hat\sigma}\right) =   \alpha(1- s(\underline{\theta})),  \ \ \ \  
F \left (\tfrac{ \hat\theta - \overline{\theta} }{\hat\sigma}\right) = \alpha s(\overline{\theta})  ,
\label{eqn:endpoints}
\end{equation}
where $F$ is the CDF of the $t_{q}$ distribution.
These equations can be solved using a zero-finding algorithm, and 
noting that $\underline \theta < \hat\theta + \hat\sigma t_{\alpha} $ and 
    $ \hat\theta + \hat\sigma t_{1-\alpha } <\overline \theta$. 
Furthermore, 
this implies that $\underline \theta < \hat\theta < \overline\theta$  
as long as $\alpha<1/2$.

The spending function 
$s(\theta)$ should be chosen on the basis of any additional information 
we may have about the value of $\theta$. 
While  Pratt's FAB interval uses  
prior information for a scalar parameter, 
in multiparameter settings such information may come 
from the data itself. 
For example, estimates of some parameters might suggest plausible 
values for others. 
In this case, we may want to use a spending function 
$\tilde s(\theta)$ 
that minimizes a Bayes risk corresponding to a ``prior'' 
distribution that is adaptively estimated  
from the data. We refer to 
such as procedure as adaptive FAB. 
Fortunately, the results of Proposition 
\ref{prp:tprop} hold not just for fixed spending functions, but 
also those that are random but statistically independent of $\hat\theta$ and 
$\hat\sigma^2$:
\begin{cor}
If $\hat\theta \sim N(\theta, \sigma^2)$ and 
$ q \hat \sigma^2/\sigma^2 \sim \chi^2_q $, and  $\hat\theta$, $\hat\sigma^2$ and  $\tilde s$ are independent, then
$\Pr( \theta\in C_{\tilde s}(\hat \theta,\hat \sigma) | \theta,\sigma )=1-\alpha$.  
\label{cor:acoverage}
\end{cor} 
This result follows by conditioning  on $\tilde s$:
$\Pr( \theta\in C_{\tilde s}(\hat \theta,\hat \sigma) | \theta,\sigma ) = 
 \Exp{ \Pr( \theta\in C_{\tilde s}(\hat \theta,\hat \sigma) |\tilde s, \theta,\sigma ) | \theta,\sigma}$, but the inner conditional probability is $1-\alpha$
since $\tilde s$ and $(\hat\theta, \hat\sigma)$ are independent
and $C_{\tilde s}$ has $1-\alpha$ coverage for each fixed $\tilde s$. 
\citet{yu_hoff_2016} made use of this fact to develop
an adaptive FAB confidence interval procedure 
for the means of multiple normal populations. 
Their adaptive procedure 
for the mean $\theta$
of a given population is $C_{\tilde s}( \hat \theta,\hat \sigma)$, 
where $\tilde s$ is 
the spending function
(\ref{eqn:sfabz}) with $(\mu,\sigma^2,\tau^2)$
replaced by estimates %$(\hat\mu,\hat \sigma^2,\hat \tau^2)$
using data from the other populations.
This procedure provides exact $1-\alpha$ confidence intervals
for each population, and is asymptotically optimal in the case
of the normal hierarchical model.

%For example, if our prior information 
%was that $\theta\sim N(\mu,\tau^2)$, we could 
%use the $w$-function of Pratt's interval given by (\ref{eqn:wfabz}). 
%When plugged-in to (\ref{eqn:fab_ci}) to form a $t$-interval, 
%the resulting procedure is not Bayes optimal 
%as it was derived to optimize a $z$-interval rather than 
%a $t$-interval. However, the it can be shown to be 
%asymptotically Bayes optimal as $q \rightarrow  \infty$, 
%as in this case $s^2\rightarrow \sigma^2$ almost surely 
%and the interval endpoints defined by the $t$-quantiles 
%converge to the corresponding $z$-quantiles. 

\section{FAB $t$-intervals for regression parameters}
We now show how the results discussed in the previous section may be 
used to construct adaptive frequentist confidence intervals 
for linear regression parameters. 
The intervals we construct have exact $1-\alpha$ constant coverage 
and do not require asymptotic approximations or  
assumptions on the design matrix or the  unknown parameters. 
Under some conditions, the intervals are also asymptotically Bayes-optimal. 
The intervals do require that the number $n$ of observations 
is larger than the number $p$ of regressors. 

\subsection{FAB regression intervals}

Consider the problem of constructing
confidence intervals for the elements of an unknown vector
$\bs \beta\in \mathbb R^p$
based on data $\bl y\in \mathbb R^{n}$ and $\bl X\in\mathbb R^{n\times p}$
from the  normal linear regression model
 $\bl y \sim N_n(\bl X \bs \beta , \sigma^2 \bl I)$.
As is well known,
\begin{align*}
\hat{\bs\beta} & = (\bl X^\top \bl X)^{-1} \bl X^\top \bl y  \sim 
  N_p( \bs \beta, \sigma^2  (\bl X^\top \bl X)^{-1} ) \\
 \hat\sigma^2\ & = || \bl y - \bl X \hat{\bs \beta}||^2/(n-p)  \sim \tfrac{\sigma^2}{n-p} \chi^2_{n-p}, 
\end{align*}
with  $\hat{\bs \beta}$ and $\hat\sigma^2$ being independent. 
In particular, 
$(\hat\beta_j - \beta_j)/(w_j \hat\sigma) \sim t_{n-p}$, where
$w_j$ is the square-root of 
the $j$th diagonal entry of
$(\bl X^\top \bl X)^{-1}$.
The UMAU confidence interval for $\beta_j$ is 
\begin{equation} 
C(\hat\beta_j,\hat \sigma) = 
 \left \{\beta_j :  \hat \beta_j + w_j\hat \sigma t_{\alpha/2} <\beta_j < 
    \hat \beta_j + w_j \hat \sigma t_{1-\alpha/2} \right  \}. 
\label{eqn:uregci}
\end{equation}
This interval  has $1-\alpha$ coverage probability and 
an expected width that is constant as a function of 
the true value of  $\beta_j$. 
%\[
%  \Pr( \beta_j \in C(\hat \beta_j, \hat \sigma) | \bs \beta,\sigma^2) =1-\alpha  \ \forall j\in \{1,\ldots, p\}, \bs\beta\in \mathbb R^p, \sigma^2\in \mathbb R^+. 
%\]

Now suppose that prior information about $\bs \beta$ %is available 
%that 
suggests that $\bs \beta \sim N_p( \bl 0 , \tau^2 \bl I) $
for some value of $\tau^2$
(other prior distributions will be discussed in Sections 4 and 5). 
 If $\tau^2$ and $\sigma^2$ were known, 
the Bayes-optimal CIP for $\beta_j$ would be obtained simply by replacing 
$\hat\theta$ and $\sigma$ in (\ref{eqn:gen_ci}) and (\ref{eqn:sfabz}) 
with $\hat\beta_j$ and $w_j \sigma$, yielding 
\begin{align}
 C_{s}(\hat\beta_j) & = \left \{ 
   \beta_j : \hat \beta_j + w_j \sigma z_{\alpha(1-s(\beta_j))} 
   < \beta_j < 
     \hat \beta_j + w_j \sigma z_{1-\alpha s(\beta_j)} \right \} \label{eqn:oracle} \\
 s(\beta_j ) & =   g^{-1}(2 w_j \sigma \beta_j/ \tau^2 ).  \nonumber
\end{align}
However, since $\tau^2$ and $\sigma^2$ are unknown we alter 
this interval as follows:
\begin{itemize}
\item $\sigma^2$ is replaced by $\hat\sigma^2$,
% the usual estimate % $\hat \sigma$ 
%that is
which is
independent of $\hat\beta_j$ and 
satisfies $(n-p) \hat \sigma^2 /\sigma^2\sim \chi^2_{n-p}$; 
\item $z$-quantiles are replaced by the quantiles of 
the $t_{n-p}$ distribution; 
\item $s(\beta_j)$  is replaced by 
   $\tilde s(\beta_j) = g^{-1}(2 w_j \tilde \sigma \beta_j/ \tilde \tau^2 )$, 
where $(\tilde \tau^2, \tilde \sigma^2)$ % are estimates of $(\tau^2,\sigma^2)$ 
%that
are independent 
of $(\hat \beta_j , \hat\sigma^2)$. 
\end{itemize}
These modifications yield an adaptive FAB interval given by 
\begin{equation} 
C_{\tilde s} (\hat\beta_j,\hat\sigma ) = 
 \left \{\beta_j :  \hat \beta_j + w_j\hat \sigma t_{\alpha (1-\tilde s(\beta_j))} <\beta_j < 
    \hat \beta_j + w_j \hat \sigma t_{1-\alpha \tilde s(\beta_j)}\right  \}. 
\label{eqn:fregci}
\end{equation}
Such an interval satisfies the conditions of 
Corollary \ref{cor:acoverage}, 
thereby guaranteeing exact $1-\alpha$ frequentist coverage, 
regardless of whether or not the values of $\bs \beta$ 
are approximately normally distributed, or if 
the estimates $\tilde \tau^2, \tilde \sigma^2$ are accurate. 
However, if the normal approximation and adaptive estimates are 
accurate, 
then we expect the resulting
FAB interval (\ref{eqn:fregci}) 
to be  close to the ``oracle'' interval (\ref{eqn:oracle}), 
which is 
Bayes-optimal and 
narrower on average than the UMAU procedure
given by (\ref{eqn:uregci}). 

The approximate optimality of $C_{\tilde s}$ 
is considered more formally in the 
next subsection using an asymptotic argument. 
First, we discuss obtaining estimators  
$(\tilde\tau^2, \tilde \sigma^2)$ 
that are independent of $(\hat\beta_j,\hat\sigma^2)$ 
so that the conditions of 
Corollary \ref{cor:acoverage} are met.  
Let $\bl P_X = \bl X (\bl X^\top \bl X)^{-1} \bl X$ 
and $\bl P_0 = \bl I -\bl P_X$ be the 
projection matrices onto the space spanned by the columns of 
$\bl X$ and the corresponding null space, respectively. 
Recall that the OLS estimate $\hat\beta_j$ is given by 
$\hat \beta_j = \bl a^\top \bl y$, where 
$\bl a$ is the $j$th row of the matrix  
$(\bl X^\top \bl X)^{-1} \bl X^\top$. 
Let $\bl P_1 = \bl a\bl a^\top/\bl a^\top \bl a$
be the projection matrix associated with $\bl a$, 
and let $\bl P_2 = \bl P_X ( \bl I - \bl P_1)$. 
We can decompose $\bl y$ as 
\begin{align*}  
\bl y = \bl I\bl y &= (\bl P_0 + \bl P_X ) \bl y \\
  &= ( \bl P_0 + \bl P_1 +\bl P_2) \bl y   \\
  &=  \bl P_0 \bl y + \bl P_1 \bl y + \bl P_2 \bl y \equiv 
  \bl y_0 + \bl y_1 +\bl y_2. 
\end{align*}
Since $\bl P_k \bl P_l =\bl 0 $ for $k\neq l$, 
we have that $\bl y_0$, $\bl y_1$ and $\bl y_2$ are 
statistically independent.  
Now 
the OLS estimate satisfies 
$\hat\beta_j = \bl a^\top \bl P_1 \bl y = \bl a^\top \bl y_1$, 
and $\hat\sigma^2 = \bl y_0^\top \bl y_0/(n-p)$, and so both 
estimates are statistically independent of each other and  the vector 
$\bl y_2$. Therefore, any estimates  $(\tilde \tau^2,\tilde \sigma^2)$
that are functions of $\bl y_2$ will be independent
of $(\hat\beta_j,\hat\sigma^2)$ and so can be used to construct  a
spending function $\tilde s(\beta_j)$ that satisfies the 
conditions of Corollary \ref{cor:acoverage}. 

To obtain such an estimate $(\tilde \tau^2,\tilde \sigma^2)$, 
let $\bl G_2$ be an orthonormal basis for the space spanned by $\bl P_2$
(for example, the matrix of eigenvectors of $\bl P_2$ that correspond 
to non-zero eigenvalues). Then $\bl G_2  \bl G_2^\top = \bl P_2$, 
$\bl G_2^\top \bl G_2 = \bl I_{p-1}$, and 
  $\bl z_2 = \bl G_2^\top \bl y_2 = \bl G_2^\top \bl y \sim N_{p-1}(
   \bl G_2^\top \bl X \bs \beta , 
\sigma^2 \bl I)$. 
Under the prior model $\bs\beta \sim N(  \bl 0 , \tau^2 \bl I)$, 
the marginal distribution for $\bl z_2$ is  therefore
\begin{equation} 
 \bl z_2  
\sim N_{p-1}( \bl 0 , \bl X_2 \bl X_2^\top \tau^2 + \sigma^2 \bl I) ,
 \label{eqn:mmod}
\end{equation}
where $\bl X_2 = \bl G^\top_2 \bl X$. A variety of empirical Bayes estimates
of $(\tau^2, \sigma^2)$
may be obtained from this marginal distribution. For example, 
noting that
$\Exp{ \bl z_2^\top \bl A^\top \bl A   \bl z_2 } = 
   \tr( \bl X_2^\top \bl A^\top \bl A \bl X_2) \tau^2 + 
     \tr(\bl A^\top \bl A ) \sigma^2$  for any matrix $\bl A$, 
unbiased moment estimates may be obtained by 
finding $(\tilde \tau^2, \tilde \sigma^2)$ that solve simultaneously 
two equations, given by 
$  \bl z_2^\top \bl A^\top \bl A   \bl z_2  = 
   \tr( \bl X_2^\top \bl A^\top \bl A \bl X_2) \tilde \tau^2 + 
     \tr(\bl A^\top \bl A ) \tilde \sigma^2 $,
for two different values of $\bl A$. 
Alternatively, $(\tilde \tau^2 ,\tilde \sigma^2)$ 
may be taken to be the maximum likelihood estimate based on 
the marginal model (\ref{eqn:mmod}). This estimate is discussed 
further in the next subsection.

To summarize, 
we have constructed statistics 
$\hat\beta_j,\hat\sigma^2,\tilde \tau^2, \tilde \sigma^2$ 
such that for each $(\bs\beta,\sigma^2)\in \mathbb R^p\times \mathbb R^+$, 
\begin{itemize}
\item[C1:]  $\hat \beta_j \sim N(\beta_j ,w_j^2 \sigma^2 ) $, 
     $(n-p)\hat \sigma^2/\sigma^2 \sim \chi^2_{n-p} $, and 
     $\hat \beta_j$ and $\hat\sigma^2$ are independent; 
\item[C2:] $(\tilde \tau^2 , \tilde \sigma^2)$ are independent of 
$( \hat\beta_j, \hat \sigma^2)$. 
\end{itemize} 
Therefore the spending function 
$\tilde s(\beta_j) = g^{-1}(2 w_j \tilde \sigma \beta_j/ \tilde \tau^2 )$ 
is independent of  $( \hat\beta_j, \hat \sigma^2)$ and so the 
 conditions of 
of Corollary \ref{cor:acoverage} are met. We summarize these results with 
the following theorem:
\begin{thm} 
If $\bl y \sim N_{n}(\bl X\bs\beta , \sigma^2 \bl I)$ 
and $(\hat\beta_j,\hat\sigma^2,\tilde \tau^2, \tilde \sigma^2)$ 
satisfy C1 and C2, 
then 
the adaptive FAB CIP
given by (\ref{eqn:fregci}) has $1-\alpha$ 
coverage for every value of $\bs\beta$ and $\sigma^2$, that is 
\[
 \Pr( \beta_j \in C_{\tilde s}(\hat\beta_j,\hat\sigma) | \bs\beta,\sigma) 
 = 1-\alpha, \  \forall ( \bs\beta , \sigma ) \in  \mathbb R^p \times \mathbb R^+. 
\]
\end{thm}

%% \sigma or \sigma^2  ?

\subsection{Approximate optimality}

As discussed above, if 
$\beta_j \sim N(0,\tau^2 )$
and $\sigma^2$ and $\tau^2$ were known then the  oracle
FAB interval $C_s$ given by  (\ref{eqn:oracle}) is 
Bayes-optimal 
in that it minimizes the prior expected interval width   
$\Exp{ | C | }$ 
among procedures $C$ that have $1-\alpha$ frequentist coverage. 
This prior expected width is an expectation over both 
the estimate $\hat\beta_j$ and the value of $\beta_j$  with 
respect to the $N(0,\tau^2 )$ prior distribution.

The adaptive FAB interval $C_{\tilde s}$ given by (\ref{eqn:fregci})
differs from the oracle FAB interval in three
ways:
the value of $\sigma^2$
has been replaced by $\hat \sigma^2$;
the $z$-quantiles
have been replaced by $t$-quantiles; and
 the spending function $s$ that depends on $(\tau^2,\sigma^2)$
has been
replaced by $\tilde s$ that depends on $(\tilde\tau^2,\tilde \sigma^2)$.
In this subsection we take $(\tilde \tau^2, \tilde \sigma^2)$ to be 
the maximizers of the likelihood given by the marginal model 
(\ref{eqn:mmod}). 
The resulting interval still has $1-\alpha$ frequentist coverage, but it is 
only an approximation to $C_s$, and so we must have 
$\Exp{ |C_{\tilde s} | } > \Exp{ | C_s|}$
since $C_s$ is Bayes-optimal.
%, that is, 
%the adaptive FAB interval has higher Bayes risk (prior expected width) than 
%$\Exp{ | C_s|}$. 
However, if 
$n-p$ is large then 
the $t$-quantiles will be close to the corresponding
$z$-quantiles, and 
we expect that $\hat\sigma^2 \approx\sigma^2$.
If $p$ is also large then
under the prior $\bs\beta \sim N(0,\tau^2 \bl I)$
we expect that $\tilde \tau^2 \approx \tau^2$ and $\tilde \sigma^2 \approx \sigma^2$. As a result,
we should have $\tilde s(\beta_j) \approx s(\beta_j)$ and so 
we expect that $\Exp{ |C_{\tilde s} | } \approx  \Exp{ | C_s|}$, 
that is, the FAB procedure will be approximately Bayes-optimal. 

We investigate this more formally with an asymptotic comparison 
of the widths of the adaptive and oracle FAB procedures. 
We first obtain an asymptotic result for a single scalar 
parameter $\beta_j$, and then discuss the result in the context of the 
linear regression model. 
Consider a sequence of experiments indexed by $n$  such that for  
each $n$ we have 
statistics $(\hat\beta_j,\hat\sigma^2,\tilde \tau^2, \tilde \sigma^2)$
that satisfy coverage conditions C1 and C2 given above.
%so that, conditional on the value of  $\beta_j$, 
%$\hat \beta_j\sim N(\beta ,w^2_j \sigma^2)$ and 
%$q\hat\sigma^2/\sigma^2 \sim \chi^2_q$ are independent 
%of each other and also of $(\tilde \tau^2, \tilde \sigma^2)$. 
Furthermore, 
suppose the following asymptotic conditions hold as 
 $n\rightarrow \infty$:
\begin{enumerate}
\item[A1.] $(n-p)\rightarrow \infty$ and  $\sigma^2/n \rightarrow \sigma_\infty^2>0$;
\item[A2.] 
%$\beta_j \sim N(0,\tau^2)$ we have
$(\tilde \tau^2, \tilde \sigma^2/n) \rightarrow 
  (\tau^2, \sigma^2_\infty)$ in probability.  
\item[A3.] $w^2_j n \rightarrow w_{0}^2>0$;
\end{enumerate}
We consider this case where $\sigma^2$ grows with $n$ since 
otherwise, if $\sigma^2$ were fixed then the 
widths of the oracle FAB, adaptive FAB and UMAU intervals 
would all converge to zero at the same rate. 

\begin{lem} 
\label{lem:aopt}
Under the conditions C1, C2, A1, A2 and A3, 
the width $|C_{\tilde s}|$ 
of the FAB procedure (\ref{eqn:fregci}) satisfies  
$\Exp{ | C_{\tilde s} | }  \rightarrow   
\Exp{ | C_{s} | }$  as $n\rightarrow \infty$, 
where $C_{s}$ is the Bayes-optimal FAB procedure for the case that 
$\hat\beta_j \sim N(\beta_j , w_0^2 \sigma_\infty^2)$ and 
$\beta_j\sim N(0,\tau^2)$. 
%, and the expectation is 
%over
%$(\hat\beta_j,\hat\sigma^2,\tilde \tau^2, \tilde \sigma^2)$ and 
%$\beta_j$ with respect to the 
%prior distribution $\beta_j\sim N(0,\tau^2)$.  
\end{lem}
A proof is in the appendix. 
The lemma says that under this asymptotic regime, the 
performance of the adaptive  FAB interval is asymptotically equivalent to 
that of the oracle FAB interval: They both have $1-\alpha$ frequentist 
coverage for each $n$, and the prior expected width of the FAB procedure 
approaches that of the oracle FAB interval as $n\rightarrow \infty$.

We now consider how this result applies 
to the linear regression model and the specific estimates  
$\hat \beta_j, \hat\sigma^2, \tilde \tau^2, \tilde \sigma^2$ described in the previous 
subsection.  Consider a sequence of experiments indexed by $n$ such that 
the following conditions hold:
\begin{enumerate}
\item[B1.] For each $n$,   
\begin{itemize}
\item $\bl X$ is full-rank;
\item $\bl y \sim N_n(\bl X \bs\beta, \sigma^2 \bl I)$  with
$\sigma^2 = n \sigma^2_\infty$; 
\item $\bs\beta \sim N(\bl 0,\tau^2 \bl I)$.   
\end{itemize}
\item[B2.] $p/n \rightarrow c \in (0,1)$ as $n\rightarrow \infty$.   
\item[B3.] The empirical distribution of the eigenvalues of 
$\bl X^\top\bl X/n$ is bounded uniformly in $n$, and converges in distribution to a non-degenerate limit as $n\rightarrow \infty$. 
%$n p /\tr(\bl X^\top \bl X) = O(1)$ and 
% $\tr( (\bl X^\top \bl X)^2 )/\tr(\bl X^\top \bl X)^2 \rightarrow 
%$ 0 $ as $n\rightarrow \infty$. 
 %some condition as $n\rightarrow \infty$.  
\end{enumerate}
%Suppose for each $n$ we have the model 
%$\bl y \sim N_n(\bl X \bs\beta, \sigma^2 \bl I)$
%with $\sigma^2 = n \sigma^2_0$, 
%and consider an asymptotic scenario where both $n$ and 
%$p$ are growing but $n$ is growing faster than $p$, 
%so that $p/n \rightarrow c \in (0,1)$ as $n\rightarrow \infty$.  
%Further assume that $\bl X$ is of full rank for each $n$.  
%The first condition B1 describes the model and prior distribution
%for each $n$, and 
%the second condition says that $p$ grows linearly with $n$. 
%The first part of B3 says that the norm of $\bl X$  does not grow 
%slower than its dimension, and the second part implies that 
%$\bl X$ is well-conditioned, i.e.\ that no one column 
%of $\bl X$ dominates the others. 

If 
conditions B1, B2 and B3 are met then 
the estimates
$(\hat\beta_j,\hat\sigma^2,\tilde \tau^2, \tilde \sigma^2)$ 
defined in Section 3.1 
satisfy the conditions  C1, C2, A1 and A2 
and so 
the FAB interval 
for  $\beta_j$ of any variable $j$  satisfying condition A3 
will satisfy the
conditions of Theorem \ref{lem:aopt}, and hence be asymptotically optimal. 
To see that this holds, first note 
that  the definition of the model in condition B1 
implies that 
$(\hat\beta_j,\hat\sigma^2,\tilde \tau^2, \tilde \sigma^2)$ 
satisfy the coverage conditions 
 C1 and C2. 
%and B3 implies that A2 is met. 
Second, asymptotic condition A1 is met by the 
 definition $\sigma^2 = n \sigma^2_\infty$ in B1
and that $n$ is growing faster 
than $p$ as assumed by  B2. 
%, and that $p\rightarrow \infty$ 
%by B2. 
The remaining necessary result is the following:
\begin{lem} 
\label{lem:mmle} 
Suppose B1, B2 and B3 hold, and that 
$(\tau^2,\sigma^2_\infty)\in \Theta$, a compact 
subset of $[0,\infty)\times (0,\infty)$. 
Let $(\tilde \tau^2, \tilde \sigma^2)$ be the maximizers over $\Theta$
of the likelihood given by the marginal model 
(\ref{eqn:mmod}).  Then  
$(\tilde\tau^2,\tilde \sigma^2/n ) \rightarrow (\tau^2 ,\sigma^2_\infty)$ 
in probability as $n\rightarrow \infty$. 
\end{lem}
This result is proven in the appendix. 
Putting Lemma \ref{lem:aopt} and Lemma \ref{lem:mmle}  gives the following 
summary of the asymptotic behavior of $C_{\tilde s}$:
\begin{thm}
Under the conditions of Lemma \ref{lem:mmle}, 
for any variable $j$ for which 
A3 holds,
the FAB interval 
$C_{\tilde s}(\hat\beta_j ,\hat\sigma^2 )$ has 
prior expected width  $\Exp{ | C_{\tilde s} | }$ that satisfies 
$\Exp{ | C_{\tilde s} | }  \rightarrow   
\Exp{ | C_{s} | }$  as $n\rightarrow \infty$,
where $C_{s}$ is the Bayes-optimal FAB procedure for the case that
$\hat\beta_j \sim N(\beta_j , w_0^2 \sigma_\infty^2)$
and $\beta_j\sim N(0,\tau^2)$.
\end{thm}
This result makes precise the heuristic idea that if $n$ and $p$  
are large, then the adaptive  FAB interval should be 
nearly as good as the oracle FAB interval.

\section{Numerical examples}

In this section we illustrate the adaptive FAB  
procedure numerically, and show how it can be modified 
to accommodate different adaptation strategies. For example, 
in the next subsection we use an empirically estimated 
prior distribution that is not centered around zero, 
thereby providing improved performance if most of 
the effects are of a common sign.  In the following subsection, 
we show how adaptation may be done separately for 
different groups of parameters, such as main effects and interactions. 
We also provide a simulation study that illustrates how a 
CIP that adapts to sparsity may have very poor coverage if the 
regression parameter is not actually sparse, whereas the 
adaptive FAB procedure maintains constant coverage for all parameter values.

\subsection{Motif regression}
\citet{conlon_etal_2003}
measured the binding intensity of a protein to each of 
$n=287$ DNA segments, and related each intensity to 
scores measuring abundance of the DNA segment in 
$p=195$ genetic motifs. These data were also used 
as an example by 
\citet{meinshausen_meier_buhlman_2009}, among others. 

Assuming a normal linear regression model for the centered and scaled data, 
the usual unbiased estimate of $\sigma^2$ is 
0.77, 
and the usual standard errors for the OLS regression 
coefficients range from 0.12 to 0.85 with a mean of 
0.30. On the other hand, 
empirical 
Bayes estimates of $\mu$ and $\tau^2$ under the prior 
 $\bs \beta \sim N_p(\mu \bl 1  , \tau^2 \bl I)$ are around 
0.004 and 0.001 respectively, 
($\hat \tau \approx  0.036$) 
suggesting that the 
true values of the elements of $\bs\beta$ are highly concentrated 
around zero. 

We constructed 95\% FAB confidence intervals for the 
effects of the $p=195$ genetic motifs, using the  adaptive FAB
procedure described in Section 3.1
except under a $N_p(\mu \bl 1 , \tau^2 \bl I ) $ distribution 
for $\bs \beta$.  This is to allow for the possibility 
that the distribution of true effects is not centered around zero, 
which seems reasonable for this particular dataset where it is 
expected that abundance has either a positive or negligible 
effect on binding intensity. 
In the analysis that follows, for each 
coefficient $j$, values of $(\tilde\mu, \tilde\tau^2, \tilde \sigma^2)$ 
are estimated from the $j$-specific vector $\bl y_2$ defined in 
Section 3.1, thereby ensuring that 
$\hat\beta_j$ is independent of  $(\tilde\mu, \tilde\tau^2, \tilde \sigma^2)$
and constant coverage of the FAB confidence interval for each $\beta_j$ is maintained. 

The intervals are shown graphically in Figure \ref{fig:motif}, 
along with the UMAU intervals for comparison. 
The FAB intervals are shorter than the UMAU intervals for 189 of the 195 effects (97\%), 
with relative widths ranging from 0.83 to 1.11, 
and being 
0.85
on average across effects. 
The number of ``significant'' effects identified by the two procedures 
is similar: twelve of the FAB CIs and eleven of the UMAU CIs do not contain 
zero. However, the two procedures identify somewhat different 
significant motifs: seven motifs are identified by both procedures, 
all with positive OLS effect
estimates. The FAB procedure identifies
an additional five motifs all with positive effect estimates, whereas 
the four additional motifs identified by the UMAU procedure 
all have negative effect estimates.

\begin{figure}
%\centerline{\includegraphics[width=6.25in]{../Replication/motif}} 
\centerline{\includegraphics[width=6.25in]{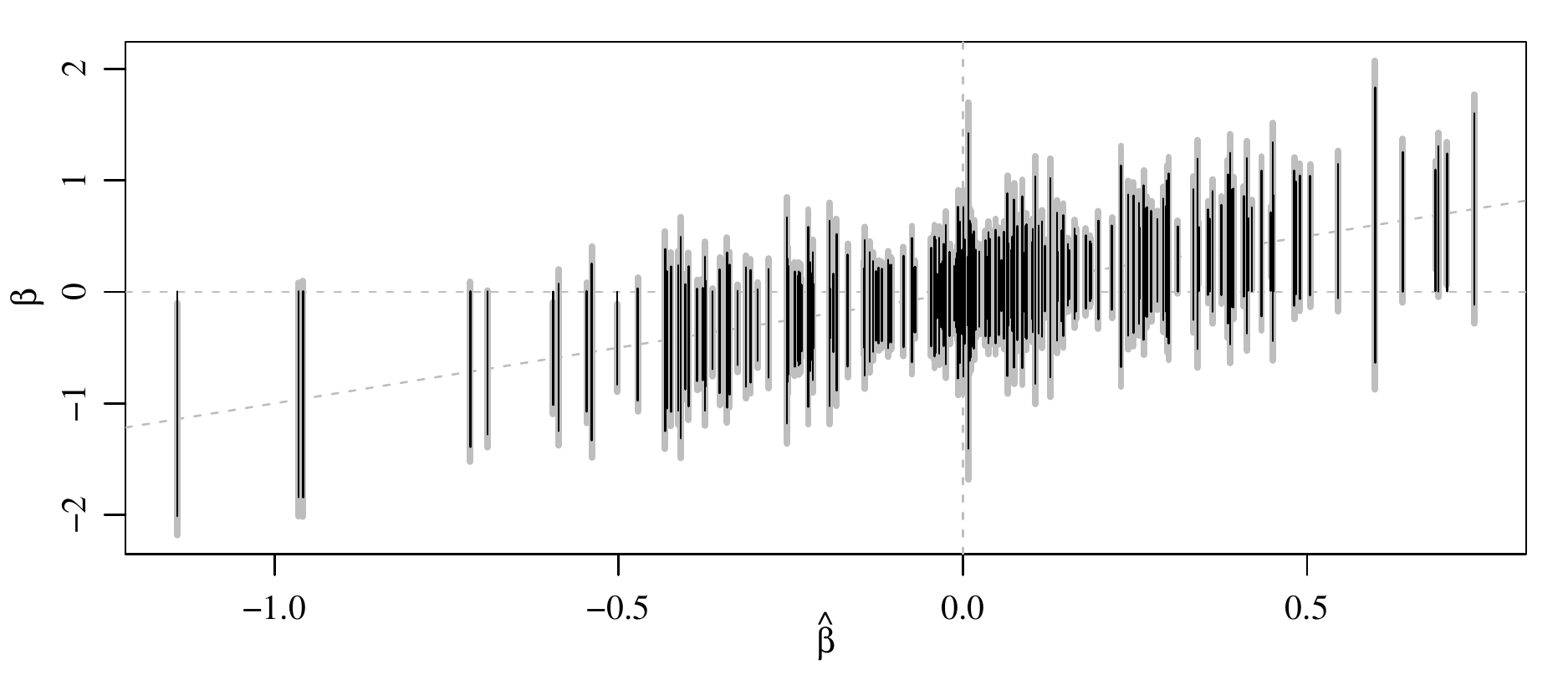}} 
\caption{ 95\% confidence intervals for  motif regression 
effects. UMAU intervals are thick gray lines, 
FAB intervals are thin black lines. }
\label{fig:motif}
\end{figure}

\subsection{Motif regression simulation study} 

\citet{zhang_zhang_2014} developed a confidence interval 
procedure for sparse parameters in high-dimensional 
normal linear regression models. 
When applied to the motif dataset, this
low dimensional projection (LDP) 
procedure produces 
intervals that are narrower than the FAB intervals for 
all regression coefficients, with 
relative widths ranging from 0.27 to 0.71,
and being
about half as wide 
on average across coefficients.
However, unlike the FAB and UMAU procedures, the actual coverage 
rates of LDP intervals  are guaranteed to 
achieve their nominal rates only  asymptotically, 
and only if certain sparsity conditions on $\bs \beta$ are met. 

To compare the performance of the UMAU, FAB and LDP procedures  we constructed 
two related simulation studies based on the motif binding dataset 
described in the previous subsection.  
In each study, we obtained  estimates $({\bs\beta}_0,\sigma^2_0)$ 
from the real data $\bl y$ and $\bl X$, and used these  estimates
to simulate  new response vectors 
$\bl y^{(k)} \sim N( \bl X {\bs \beta}_0 ,\sigma^2_0 \bl I)$ 
independently for $k=1,\ldots, 5000$.  
For each response vector $\bl y^{(k)}$ we construct 
UMAU,  FAB and LDP confidence intervals for each 
of the $p=195$ regression coefficients. These intervals are used to 
obtain Monte Carlo approximations to the finite-sample 
coverage rates of the LDP procedure, as well as 
approximations to the 
expected interval widths of the UMAU, FAB  and LDP procedures. 

In the first of these two simulation studies 
we simulated 5000 datasets from the model 
$\bl y^{(k)} \sim N_n( \bl X \bs \beta_0, \sigma^2_0 \bl I)$, 
where $\bl X$ is the original design matrix 
and  $\bs\beta_0$ is the lasso 
estimate from the original data, using an empirical Bayes
estimate of the $L_1$-penalty parameter.  
This resulted in
a sparse ${\bs \beta}_0$-vector with  176  of the 195 coefficients
being identically zero, so in the context of this simulation study, 
the ``truth'' is highly sparse. 
The value of $\sigma^2_0$ used to simulate the data 
was the usual unbiased estimate 
from the original data. 
We computed the UMAU, FAB and LDP confidence intervals
for each of the 
5000 simulated datasets.
The  widths  of the 
FAB and LDP intervals were 85\% and  43\% of the UMAU interval widths
respectively, on average across datasets and parameters. 
The empirical coverage rates of the
nominal 95\% LDP intervals ranged between 93.8 and 96.1 percent.
There was some evidence that the coverage rates were not 
exactly 95\%: Exact level-.05 binomial tests rejected the hypothesis 
that the coverage rates were 95\% for 71 of the 195 
regression parameters (36\%). All of these 71 parameters
had true values of 0, and the empirical coverage rates 
of 64 of these 71 parameters were larger than 95\%,
suggesting that LDP intervals slightly overcover 
$\beta_j$ when it is zero. However, in general  the coverage 
rates of the LDP procedure were very close to the nominal rates, 
in this case where the truth is sparse.

The second simulation study was the same as the first except the 
value $\bs\beta_0$ used to generate the simulated data was the OLS estimate 
from the original data, and so in this case the ``true'' 
regression model is not sparse.
On average across the 5000 simulated datasets and 195 parameters, 
the widths  of the
FAB and LDP intervals were 88\% and  54\% of the UMAU interval widths
respectively, similar to the results from the first study.  
These 
relative widths are shown in the left panel of Figure \ref{fig:motif_sim}. 
However, the coverage rates for the LDP intervals were generally
far from their nominal levels: Based on  exact binomial tests, 
coverage rates for 
183 of the 195 parameters were significantly different from 95\% 
(at level 0.05). 
As shown in the right panel of Figure \ref{fig:motif_sim}, 
the LDP intervals generally overcover parameter values near zero, and 
greatly undercover parameters larger in magnitude.
For comparison, the empirical coverage rates of the FAB 
intervals are also shown. These rates show no evidence of 
deviation from the nominal rates, as should be the case - the FAB 
intervals have exact 95\% coverage for each component of $\bs\beta$ 
by construction.

\begin{figure}
%\centerline{\includegraphics[width=6.5in]{../Replication/motif_sim}}
\centerline{\includegraphics[width=6.5in]{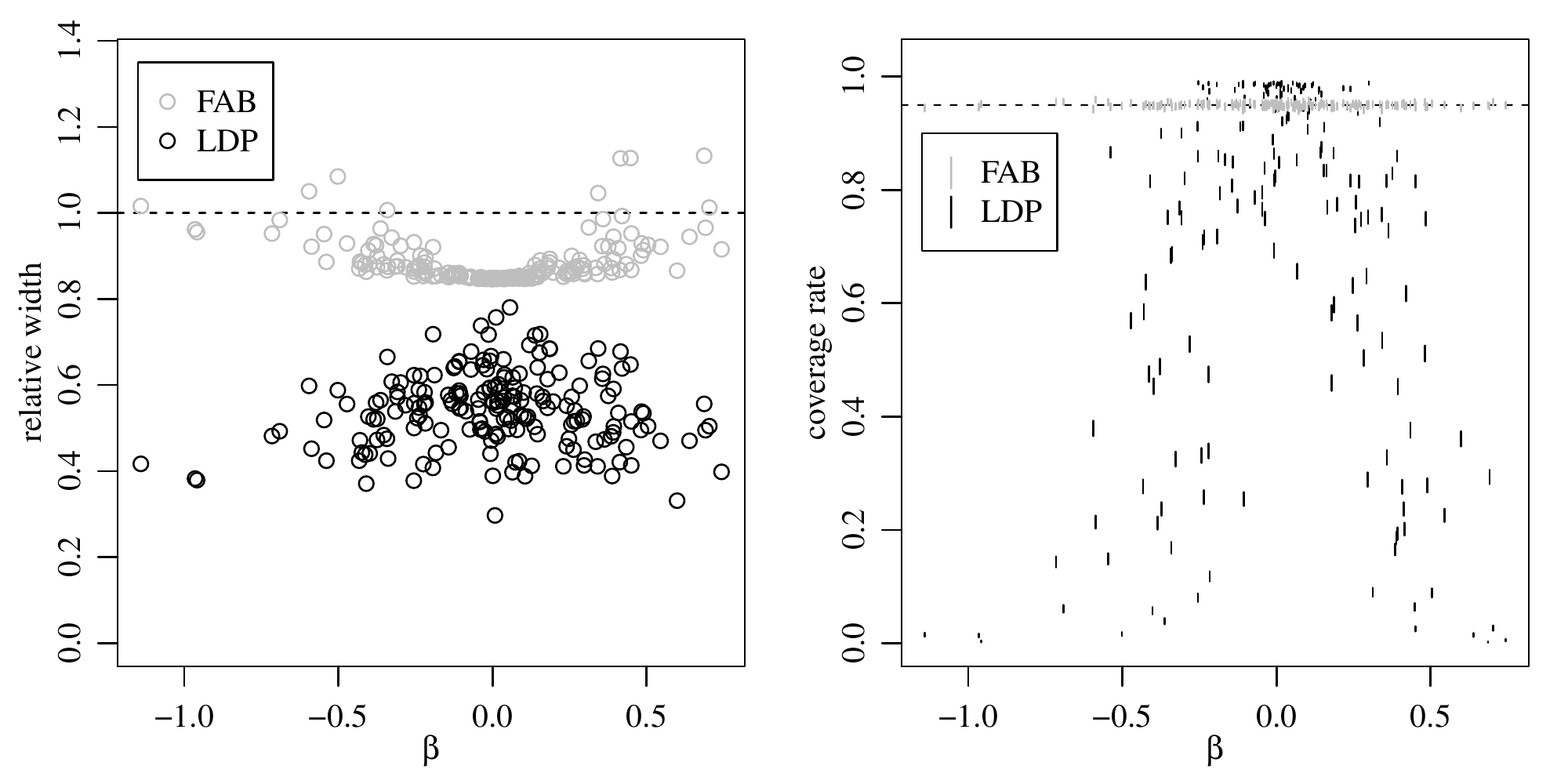}}
\caption{Relative interval widths (left panel) and coverage rates (right panel)
for the motif simulation study with  non-sparse $\bs\beta$. Error bars for the coverage rates are
 Clopper-Pearson 95\% intervals.    }
\label{fig:motif_sim}
\end{figure}

\subsection{Diabetes progression} 

\citet{efron_etal_2004} 
considered parameter estimation for a model of 
diabetes progression from data on 
ten explanatory variables from each of 
$n=442$ subjects.
The expected progression of a subject was assumed to be 
a linear function of the linear, quadratic and two-way 
interaction effects of the ten variables, resulting 
in a linear model with 
$p=64$  regressors total (the binary sex variable does not 
have a separate quadratic effect). 

We generally expect that main effects will be larger than 
quadratic effects and two-way interactions. For this 
reason, it makes sense to obtain adaptive intervals separately 
for these three types of parameters, so that the 
spending function $\tilde s$  used to obtain the confidence interval 
for  the effect of a given regressor is obtained adaptively 
from the estimated effects of regressors in the same category. 
This can easily be done as follows: Write the design matrix 
$\bl X $ as $\bl X  = [  \bl X_1 , \bl X_2 , \bl X_3 ]$, 
where $\bl X_1$, $\bl X_2$, $\bl X_3$ are the design matrices 
corresponding to the main effects, quadratic effects and 
two-way interactions, respectively, and 
let 
$\bs \beta^\top = [ \bs\beta^\top_1 , \bs\beta^\top_2 , \bs\beta^\top_3 ]$
be the corresponding partition of $\bs \beta$. 
To obtain the FAB CIs for the main effects, we 
let $\bl G$ be an orthonormal basis for the null space of 
$[ \bl X_2  , \bl X_3 ]$. Letting 
$\tilde {\bl y} = \bl G^\top \bl y$ and $\tilde {\bl X} = \bl G^\top \bl X$, we have
$\tilde{ \bl y} 
\sim N_{n-p_2-p_3} ( \tilde{\bl X} \bs\beta_1 , \sigma^2 \bl I) 
$. 
We can then apply the FAB CI procedure to $(\tilde{ \bl y}, \tilde{\bl X})$
to obtain intervals that adapt to the magnitude of $\bs\beta_1$ 
(and not to the magnitudes of $\bs\beta_2$ and $\bs\beta_3$). 
Adaptive confidence intervals for 
$\bs\beta_2$ and $\bs\beta_3$ can be obtained analogously. 

In the analysis that follows we use
an adaptively estimated  $N(0,\tau^2)$ prior distribution 
for each coefficient. 
Recall that our FAB procedure generates an  empirical Bayes estimate  $\tilde \tau^2$
of $\tau^2 =\Var{\beta_j}$ for each coefficient $j$ 
that is statistically independent of the OLS estimate 
$\hat\beta_j$.  For the main effects the values of $\tilde\tau$ ranged 
between 0.19 and 0.21, with a mean of 0.20, and were larger than 
the standard errors of the OLS  coefficients except for those of four somewhat 
co-linear predictors. 
In contrast, values of $\tilde \tau$ for the 
quadratic and interaction terms were all less than 0.03, 
and were all less than the corresponding standard errors. 

We computed the adaptive FAB interval 
for each regression coefficient using these coefficient-specific 
estimates of $\tau^2$. 
%and assuming a normal prior distribution for each coefficient. 
The FAB intervals are as narrow or narrower than all but three of the 
corresponding UMAU intervals, with the relative interval 
widths ranging from 0.84 to 1.0003, and being 0.86 on average. 
The FAB CIs for the main effects are essentially the same 
as the UMAU CIs, whereas the FAB CIs for the 
quadratic and interaction terms are all narrower than 
the corresponding UMAU intervals, by about 16\% on average. 
This example illustrates some flexibility of the FAB procedure, 
in that the adaptation for a particular parameter may be 
based on a subset of the data information that is deemed 
most relevant for that parameter.

% For a given $\beta_j$ we have $n-p-1 = 91$ dof 
% to estimate hyperparams and the value of $s$. 
% I took 46 of these for $s$, and the remaining for 
% the estimate of $\sigma^2$ that is used in the $w$ 
% function. 

\section{Discussion}
We have constructed  a class of $1-\alpha$ confidence interval
procedures (CIPs) for individual regression coefficients  of 
the normal regression model $\bl y \sim N_n(\bl X \bs\beta ,
 \sigma^2 \bl I )$. Each member of this class corresponds to 
a spending function $s:\mathbb R\rightarrow [0,1]$.  
Under the regression model, every member of the 
class has constant $1-\alpha$ coverage  
for all possible values of $\bs \beta$, $\sigma^2$ and 
full-rank design matrices $\bl X$.  
We have described 
a method of adaptively selecting 
the spending function 
so that the across-parameter average interval width is 
reduced, and the $1-\alpha$ coverage rate is 
maintained for each regression coefficient. 
The coverage guarantee is 
non-asymptotic,  does not rely on $\bs\beta$ being sparse 
and does not rely on conditions on the design matrix. 
However, under some assumptions on the distribution 
of the elements of $\bs\beta$ and the design matrix, 
the adaptive technique we propose is asymptotically optimal 
as both $n$ and $p$ increase. 

The spending function $s(\beta)$ that we adaptively 
estimate from the data is based on 
a normal prior distribution 
for the elements of $\bs\beta$. As 
such, we expect our procedure to provide the most improvement 
when the empirical distribution of $\beta_1,\ldots, \beta_p$ 
is approximately normal. 
If instead we suspect that $\bs\beta$ is sparse, it may seem 
preferable to 
base the adaptation 
on other families of prior distributions, 
such as Laplace  or ``spike and slab'' distributions. 
Some numerical work not presented here suggests 
that FAB intervals obtained using the  Laplace  
family of priors are in practice similar to those 
obtained with normal priors. 
However, FAB procedures based on spike and slab priors 
do seem more efficient but also present a problem: 
The spending function for a spike and slab prior 
is not generally nondecreasing, and so by Theorem 1 the 
corresponding confidence region may not be an 
interval. We suspect that non-interval confidence regions 
have limited appeal in practice, but even if they 
were of interest they 
present the numerical challenge of identifying 
multiple disconnected sets of parameter values 
to include in the region. 

We conjecture that the LDP interval procedure 
proposed by \citet{zhang_zhang_2014} 
could be related to a FAB procedure based on a sparsity-inducing 
prior distribution, as both procedures should be asymptotically 
optimal under a sparse regime. The fact that a FAB procedure 
based on such a prior might produce non-interval confidence 
regions might partly explain why the LDP procedure fails when 
$\bs \beta$ is non-sparse: If a sparse $1-\alpha$ FAB procedure 
yields a non-interval region but only  a sub-interval is 
numerically identified, 
then the coverage rate will be below $1-\alpha$.

\medskip

Adaptive FAB intervals for linear regression coefficients  
may be computed using the {\sf R}-package
{\tt fabCI}.
Complete replication code for the numerical examples in this article
is available at the first author's website. 
This research was partially supported by NSF grant DMS-1505136.

\appendix

\section*{Proofs}

\subsection*{Proof of Theorem 1}
Items 1 and 2 of the theorem were shown in \citet{yu_hoff_2016}. 
To prove Item 3, suppose $s(\theta)$ is not nondecreasing so 
that there exists $\theta_1<\theta_2$ with $s(\theta_1) >s(\theta_2)$. 
We will show that there are a range of values of $(\hat\theta,\hat\sigma)$
with $\hat\theta<\theta_1$
for which $\theta_2$ (and $\hat \theta$) are in $C_s(\hat\theta,\hat\sigma)$ 
but $\theta_1$ is not. 
Let 
$\underline t_j = t_{\alpha (1-s(\theta_j))}$ and 
$\bar t_j = t_{1-\alpha s(\theta_j)}$
for $j=1,2$.   
Both $t_{\alpha(1- s)}$ and $t_{1-\alpha s}$ are decreasing 
in $s$  so 
$\underline t_1 < \underline t_2 < \bar t_1 <\bar t_2$. 
%Since $\hat\theta< \theta_1 <\theta_2$
%we have 
%$\theta_j> \hat\theta + \hat\sigma \underline t_j$ for $j=1,2$. 
For $\theta_2$ to be in the confidence region and 
$\theta_1$ not to be, we need  
$\hat\theta + \hat\sigma \underline t_2 < \theta_2 <
 \hat\theta + \hat\sigma \bar t_2$ and 
$\theta_1> \hat\theta + \hat\sigma \bar t_1$, or equivalently
\begin{align}
\bar t_1  & < (\theta_1 -\hat\theta)/\hat\sigma  \label{eqn:t1b} \\
 \underline t_2  & <  (\theta_2 - \hat\theta)/\hat\sigma < \bar t_2.  \label{eqn:t2b}
\end{align}
The set of values  $(\hat\theta ,\hat\sigma)$ for which this holds 
has positive Lebesgue measure on $(-\infty, \theta_1) \times (0,\infty)$. 
For example,  both $(\theta_1 -\hat\theta)/\hat\sigma$ and 
$(\theta_2 -\hat\theta)/\hat\sigma$ can be made 
simultaneously arbitrarily close to 
any number between $\bar t_1 \wedge \underline t_2$ and $\bar t_2$
by taking 
$\hat\sigma$ and $-\hat\theta$ to be sufficiently  large. 
The probability of observing values of $(\hat\theta,\hat\sigma)$ 
satisfying (\ref{eqn:t1b}) and (\ref{eqn:t2b})
that yield a  non-interval confidence region
is therefore  greater than zero.

\subsection*{Proof of Lemma 1}
For notational convenience, in this proof 
we write $\beta_j$ and $w_j$ as $\beta$ and $w$, and 
write the $\alpha$-quantile of the $t_q$ distribution as $t(\alpha)$, 
suppressing the index that denotes the degrees of freedom.  
We begin the proof of Lemma 2 with another lemma:
\begin{lem}
        The width $|C_{\tilde s}|$ of $C_{\tilde s}$ satisfies
                $|C_{\tilde s} |  < |\hat \beta| +  w \hat \sigma(  |t(\alpha/2)|+|t(1-\alpha/2)|) $. 
%        where $t$-quantiles are those of the $t_q$-distribution.
\end{lem}

\begin{proof}
Recall that the endpoints $\underline \beta$ and $\bar \beta$ of $C_{\tilde s}$ are solutions to
%        \begin{equation} %\label{eq:endpoints}
%                \begin{array}{l}
%                        \bar \beta - w \hat \sigma t(1- \alpha \tilde s(\bar \beta) )=\hat \beta  \\
%                        \underline \beta - w \hat \sigma t(\alpha (1-\tilde s(\underline \beta))) =\hat \beta.\\
%                \end{array}
%        \end{equation} 
\begin{align*}
 \hat\beta & = \bar \beta - w \hat \sigma t(1- \alpha \tilde s(\bar \beta) ) \\
\hat \beta & =   \underline \beta - w \hat \sigma t(\alpha (1-\tilde s(\underline \beta))) 
\end{align*}
        Here $\tilde s (\beta)$ is defined as $
        \tilde s(\beta)  = g^{-1} ( 2w \tilde \sigma \beta/ {\tilde \tau}^2 )$, where $g(s)  =  \Phi^{-1}(\alpha s) - \Phi^{-1}(\alpha(1-s) )$.
        At the upper endpoint, we have $\tilde s(\bar \beta) = F((\hat \beta-\bar \beta)/(w \hat \sigma))   /\alpha$,
        where $F$ is the CDF of the $t_q$-distribution.
        When $\bar \beta > 0$, we have $\tilde s(\bar \beta) > g^{-1}(0) = 1/2$. Thus $\bar \beta < \hat \beta - w \hat \sigma t(\alpha/2)$. Also, $g^{-1} ( 2w \tilde \sigma \bar \beta/ {\tilde \tau}^2 ) < 1$, so $\bar \beta > \hat \beta - w \hat \sigma t(\alpha)$. When $\bar \beta < 0$, $\hat \beta - w \hat \sigma t(\alpha/2)< \bar \beta$. This implies that
        \begin{align*}
                &\hat \beta - w \hat \sigma t( \alpha)< \bar \beta < \hat \beta  - w \hat \sigma t( \alpha /2 )          & \ \text{if} \ \bar \beta > 0\\
                &\hat \beta - w \hat \sigma t( \alpha /2)<  \bar \beta <0         & \  \text{if} \  \bar \beta < 0. 
        \end{align*}
        Similarly we have
        \begin{align*}
                & 0 <\underline \beta < \hat \beta - w \hat \sigma t(1 - \alpha /2 )          & \ \text{if} \   \underline \beta  > 0\\
                & \hat \beta - w \hat \sigma t(1 - \alpha /2)<\underline \beta < \hat \beta - w \hat \sigma t(1 - \alpha )          & \ \text{if} \  \underline \beta < 0. 
        \end{align*}
        Therefore
        \begin{equation*}
                |C_{\tilde s}| =\bar \beta - \underline \beta < |\hat \beta| + w \hat \sigma(  |t(\alpha/2)|+|t(1-\alpha/2)|).
        \end{equation*}
\end{proof}

Now we prove Lemma 1. We denote the endpoints of the oracle CIP $C_s$ as $\bar \beta$ and $\underline \beta$, which are the solutions to
        \begin{equation*} 
                \begin{array}{l}
                        \bar \beta - w_0 \sigma_\infty \Phi^{-1}(1- \alpha s(\bar \beta) )=\hat \beta  \\
                        \underline \beta - w_0 \sigma_\infty \Phi^{-1}( \alpha (1-s(\underline \beta))) =\hat \beta.\\
                \end{array}
        \end{equation*}
 We denote the endpoints of $C_{\tilde s}$ as $\bar \beta^n$ and $\underline \beta^n$, which  are the solutions to
        \begin{equation*} 
                \begin{array}{l}
                        \bar \beta^n - w \hat \sigma t( 1-\alpha \tilde s(\bar \beta^n) )=\hat \beta^n  \\
                        \underline \beta^n - w \hat \sigma t( \alpha (1-\tilde s(\underline \beta^n))) =\hat \beta^n.\\
                \end{array}
        \end{equation*}
        We first prove that $|C_{\tilde s}|-|C_{s}| = (\bar \beta^n - \bar \beta) + (\underline \beta - \underline \beta^n) \stackrel{p}{\rightarrow } 0$ as $n \to \infty$ for each fixed $\hat \beta$.
        We can write the upper endpoints as $\bar \beta^n = G_n(w \hat \sigma, w\tilde \sigma, \tilde \tau^2, \hat \beta^n)$, and $\bar \beta = G(w_0\sigma_\infty, w_0\sigma_\infty, \tau^2, \hat \beta)$, where $G$ and $G_n$ are continuous functions of their parameters. The functions $G$ and $G_n$ are different in that the former is based on $z$-quantiles, while the latter uses $t$-quantiles.
        We have
        \begin{equation}
        \begin{split}
                |\bar \beta^n - \bar \beta| &=|G_n(w \hat \sigma, w \tilde \sigma, \tilde \tau^2, \hat \beta^n) - G(w_0\sigma_\infty, w_0\sigma_\infty, \tau^2, \hat \beta)| \\ 
                & \leq |G_n(w \hat \sigma, w \tilde \sigma, \tilde \tau^2, \hat \beta^n) - G(w \hat \sigma, w \tilde \sigma, \tilde \tau^2, \hat \beta^n)|  \\ & + |G(w \hat \sigma, w \tilde \sigma, \tilde \tau^2, \hat \beta^n) - G(w_0 \sigma_\infty, w_0 \sigma_\infty, \tau^2, \hat \beta)|.        
                \label{twocomp}
        \end{split}  
        \end{equation}
        The second term converges to zero in probability because
        $(w \hat \sigma, w \tilde \sigma, \tilde \tau^2, \hat \beta^n) \stackrel{p}{\rightarrow } (w_0 \sigma_\infty,w_0 \sigma_\infty,\tau^2, \hat \beta)$.
        Elaborating on the convergence of the first term,
        note that $G_n$ is a monotone sequence of continuous functions:
        Given $G_{n1}$ and $G_{n2}$ where $n_2 > n_1$, and suppose the corresponding degrees-of-freedom of the $t$-quantiles are $q_1$ and $q_2$. We have $q_2 \geq q_1$, thus $t_{q_2}(1-\alpha \tilde s) \leq t_{q_1}(1- \alpha \tilde s)$. Hence $\hat \beta^n - w \hat \sigma t_{q_2}(1-\alpha \tilde s(\hat \beta^n)) \geq \hat \beta^n - w \hat \sigma t_{q_1}(1-\alpha \tilde s(\hat \beta^n))$. Therefore $G_{q_2}(w \hat \sigma^2, w \tilde \sigma^2, \tilde \tau^2, \hat \beta^n) \leq G_{q_1}(w \hat \sigma^2, w \tilde \sigma^2, \tilde \tau^2, \hat \beta^n)$, and so
        by Dini's theorem, $G_n \rightarrow  G$ uniformly on
        a compact set of $(s^2 \hat \sigma^2, s^2 \tilde \sigma^2, \tilde \tau^2, \hat \beta^n)$ values.
        Since
        $(s^2 \hat \sigma^2, s^2 \tilde \sigma^2, \tilde \tau^2, \hat \beta^n) \stackrel{p}{\rightarrow } (v^2\sigma_\infty^2, v^2\sigma_\infty^2, \tau^2, \hat \beta)$,
        for arbitrary $\epsilon > 0$ and $\delta >0$, there exists a number $N(\epsilon, \delta)$ such that when $n > N(\epsilon, \delta)$,
         $|w \hat \sigma - w_0\sigma_\infty| \leq \delta$, $|w \tilde \sigma - w_0\sigma_\infty| \leq \delta$, $|\tilde \tau^2 - \tau^2| \leq \delta $ and $|\hat \beta^n - \hat \beta| \leq \delta $ with probability at least $1-\epsilon$. Therefore, for arbitrary $\eta >0$,
         \begin{equation*}
         \lim_{n\to \infty} P(|G_n(w \hat \sigma, w \tilde \sigma, \tilde \tau^2, \hat \beta^n) - G(w \hat \sigma, w \tilde \sigma, \tilde \tau^2, \hat \beta)| < \eta ) > 1- \epsilon. 
         \end{equation*}
        Since $\epsilon$ is arbitrary, we conclude that the first
        term in (\ref{twocomp}) converges to zero in probability.

 Now we show the expected width converges to the oracle width by integrating over $\hat \beta$.
        This is done by first showing $| C_{\tilde s}|$
        is uniformly integrable and then applying Vitali's theorem.
        By the previous lemma we know that
        \begin{equation*}
                        |C_{\tilde s}|  < |\hat \beta^n| +  w \hat \sigma(  |t(\alpha/2)|+|t(1-\alpha/2)|).
                \label{bound2}
        \end{equation*}
        Note that $|t(\alpha/2)|+|t(1-\alpha/2)| <  |t_1(\alpha/2)|+|t_1(1-\alpha/2)| = c_1 < \infty$, where $t_1$ is the $t$-quantile with one degree of freedom. We now show $|C_{\tilde s} |$ is $L_2$-bounded. We have
        \begin{equation*}
                |C_{\tilde s}|^2 < |\hat \beta^n|^2 +  c_1^2 w^2 \hat \sigma^2 + 2 |\hat \beta^n| c_1 w \hat \sigma.
        \end{equation*}
        Here $E[|\hat \beta^n|^2] = \beta^2 + w^2\sigma^2$ and $E[w^2 \hat \sigma^2] =w^2 \sigma^2$. Since $w^2 \sigma^2 \to w_0^2\sigma_\infty^2 < \infty$, thus $E[|\hat \beta^n|^2]$ and $E[w^2 \hat \sigma^2]$ are both bounded for all $n$. Similarly, $E[w \hat \sigma]$ and $E[|\hat \beta^n|]$ are also bounded. Therefore, it is easy to see that $|C_{\tilde s}|$ is $L_2$-bounded, which implies that  $|C_{\tilde s}|$ is uniformly integrable. By Vitali's Theorem,
        $
        \lim_{n\rightarrow \infty} \Exp{ |  C_{\tilde s}|} =
        \Exp{ |  C_{s}|} = 0$.

\subsection*{Proof of Lemma 2} 
We first prove a consistency result for 
a notationally simpler model, and then discuss how the result
applies to the marginal model 
(\ref{eqn:mmod}). The simpler model is 
$\bl y \sim N_n(\bl 0 , \Lambda \tau^2 + \bl I \sigma^2)$ 
where $\Lambda$ is a known  diagonal matrix with positive entries. 
Let  $\theta_0 = (\tau^2_0,\sigma^2_0)$ be the true value
of 
$\theta=(\tau^2,\sigma^2)$, and 
let 
$q_{i}(\theta) =  y_i^2/(\lambda_i \tau^2 + \sigma^2) - 
 \log \frac{ \lambda_i \tau^2_0 + \sigma^2_0  }{ \lambda_i \tau^2 + \sigma^2}$, 
which is -2 times the contribution of $y_i$ to the log likelihood 
plus a constant. 
The MLE is therefore given by 
$\hat\theta = \arg \min  Q_n(\theta)$, where
$Q_n(\theta) = \sum_{i=1}^n q_i(\theta)/n$. 
%\begin{align*}
% Q_n(\theta) & = \frac{1}{n} \sum_{i=1}^n  
%   q_{i}(\theta)   \\
%&= \frac{1}{n} \sum_{i=1}^n \left  ( 
% \log (\lambda_i \tau^2 + \sigma^2) + 
%  y_i^2/(\lambda_i \tau^2 + \sigma^2)  \right ) .
%\end{align*}
\begin{lem}
Assume $\theta_0 \in \Theta$, a compact subset of 
$[0,\infty)\times (0,\infty)$, and let 
$\hat\theta_n = \arg \min_{\Theta} Q_n(\theta)$. 
For each $n$ assume that  $F_n$, 
the empirical distribution of  the diagonal entries of $\Lambda$, 
has support on $[0,\bar \lambda]$, where $\bar\lambda$ is 
fixed for all $n$. 
If 
$F_n$ converges weakly to a nondegenerate distribution $F_0$ as
 $n\rightarrow \infty$ 
then $\hat \theta \stackrel{p}{\rightarrow} \theta_0$ as 
$n\rightarrow \infty$. 
\end{lem}
We prove consistency of $\hat \theta$ in three steps: 
First, we show that
$Q_n(\theta) - \Exp{Q_n(\theta)}$ converges uniformly to  
zero as $n\rightarrow \infty$. Second, we show that this implies that 
as a function of $\theta\in \Theta$, 
$Q_n$ converges uniformly to a function 
$Q_0$. Third, we show that $Q_0$ is uniquely minimized at  
$\theta_0$.  Consistency of $\hat \theta$ follows from these latter two results
(see, for example, Theorem 2.1 of \citet{newey_mcfadden_1994}). 

%Let 
%$\tilde Q_n(\theta) = \Exp{ Q_n(\theta) } $ where the expectation is 
%over $\bl y$, so that 
% \[
%\tilde Q_n(\theta) = \frac{1}{n}\sum_{i=1}^n  \left ( 
%\frac{\lambda_i \tau_0^2 + \sigma_0^2}{\lambda_i\tau^2 + \sigma^2} -
%\log \frac{\lambda_i \tau_0^2 +\sigma_0^2}{\lambda_i\tau^2 + \sigma^2} \right).
%\]
For each $n$ the expectation of $Q_n(\theta) - \Exp{ Q_n(\theta)}$ is 
zero and the variance is given by 
\[
  \Var{ Q_n(\theta) } = 
  \frac{2}{n} \left ( \frac{1}{n} \sum_{i=1}^n  (  \frac{\lambda_i\tau^2_0 + \sigma^2_0 }{ \lambda_i \tau^2 + \sigma^2 } )^2 \right )  = 
   \frac{2}{n} \left (  \int  (  \frac{\lambda\tau^2_0 + \sigma^2_0 }{ \lambda \tau^2 + \sigma^2 }  )^2 \, dF_n(\lambda)   \right ). 
\]
%where $F_n$ is the empirical distribution of the diagonal entries of 
%$\Lambda$.  
The integrand is a bounded function of $\lambda$ 
on any bounded set $[0,\bar \lambda] $ as long as 
$\sigma^2>0$. 
Therefore, if $F_n$ converges weakly to a distribution $F_0$
with such bounded support, then
for all  $\theta=(\tau^2,\sigma^2) \in [0,\infty)\times (0,\infty)$
the integral in the parentheses converges to a finite limit and
the variance of $Q_n(\theta) - \Exp{ Q_n(\theta)}$ converges to zero.
Thus $Q_n(\theta) - \Exp{ Q_n(\theta)} \stackrel{p}{\rightarrow } 0 $
as $n\rightarrow \infty$. To show that this convergence is uniform 
we use Theorem 3 of \citet{andrews_1992}. 
Using basic calculus, it can be shown that the functions $q_i(\theta)$ 
satisfy the Lipschitz condition 
$
 |q_i(\theta ) - q_i(\theta') | < 
g(y_i,\lambda_i) \times  || \theta -\theta'||  
$
where $g(y_i,\lambda_i) = (1+y_i^2/\epsilon) (\lambda_i+ 1)/\epsilon$ 
and $\epsilon=\min_\Theta \sigma^2$. 
Andrews' 
Theorem 5 says that if (i) $\Theta$ is compact,  (ii) $Q_n\stackrel{p}{\rightarrow}0$
for each $\theta\in \Theta$,  and  (iii)
$\sup_{n\geq 1} \sum \Exp{ g(y_i,\lambda_i) } /n  <\infty$, 
then $\sup_\Theta | Q_n(\theta) |\stackrel{p}{\rightarrow} 0$. 
This last condition holds in this case because 
$\Exp{ g(y_i,\lambda_i)}$ is quadratic in $\lambda_i$, the values of 
which are all bounded by assumption. Therefore, 
$Q_n(\theta)$ converges uniformly in probability  to zero as $n\rightarrow 
\infty$. 

Now consider the limiting value of $\Exp{ Q_n(\theta)}$. We have
\begin{align*}
\lim_{n\rightarrow \infty} \Exp{ Q_n(\theta)}  & =
\lim_{n\rightarrow \infty} \int 
  \left ( 
\frac{\lambda \tau_0^2 + \sigma_0^2}{\lambda\tau^2 + \sigma^2} -
\log \frac{\lambda \tau_0^2 + \sigma_0^2}{\lambda\tau^2 + \sigma^2} \right) 
 \, dF_n(\lambda)  \\ 
 &  = 
  \int 
  \left ( 
\frac{\lambda \tau_0^2 + \sigma_0^2}{\lambda\tau^2 + \sigma^2} -
\log \frac{\lambda \tau_0^2 + \sigma_0^2}{\lambda\tau^2 + \sigma^2} \right) 
 \, dF_0(\lambda) \equiv Q_0(\theta), 
\end{align*}
for all values of $\theta$ for which the integrand is a bounded
function of $\lambda$. On $\lambda\in [0,\bar\lambda]$,  the ratio
$(\lambda\tau^2_0 + \sigma_0^2)/(\lambda \tau^2 + \sigma^2)$ is
bounded between $\min\{ \tfrac{ \bar \lambda \tau^2_0 + \sigma_0^2}{ \bar \lambda \tau^2 + \sigma^2}, \tfrac{\sigma^2_0}{\sigma^2}\}$
and 
$\max\{ \tfrac{ \bar \lambda \tau^2_0 + \sigma_0^2}{ \bar \lambda \tau^2 + \sigma^2}, \tfrac{\sigma^2_0}{\sigma^2}\}$, 
and so the 
integrand is bounded in $\lambda\in [0,\bar \lambda]$ for all
$(\tau^2,\sigma^2) \in [0,\infty)\times (0,\infty)$.  
Furthermore, the convergence of $\Exp{Q_n(\theta)} $ to $Q_0(\theta)$ 
is uniform, since the integrand is bounded as a function of $\theta$ 
on the compact set $\Theta$ (\citet{rangarao_1962}). 
Together with the uniform convergence in probability of $Q_n(\theta)- \Exp{Q_n(\theta)}$ to zero, this implies uniform convergence in probability of $Q_n$ to $Q_0$. 

Finally we show that $Q_0(\theta)$ has a unique minimizing 
value at $\theta_0$. % if $\theta_0\in (0,\infty)\times (0,\infty)$. 
After computing the gradient of $Q_0(\theta)$, it is easily shown 
that a critical point $(\tau^2,\sigma^2)$ must satisfy 
\begin{align*}
 \int \lambda^k \frac{ \lambda \tau^2 + \sigma^2 }{(\lambda \tau^2 +\sigma^2)^2}  \, dF_0(\lambda)  = 
 \int \lambda^k \frac{ \lambda \tau^2_0 + \sigma^2_0 }{(\lambda \tau^2 +\sigma^2)^2}  \, dF_0(\lambda) 
\end{align*}
for $k\in\{0,1\}$. 
Rearranging, it can be shown that a critical point satisfies
\begin{equation}
\begin{pmatrix}
m_1 & 1 \\
m_2 & m_1 
\end{pmatrix}
\begin{pmatrix} \tau^2 \\ \sigma^2 \end{pmatrix} =
\begin{pmatrix} 
m_1 & 1 \\
m_2 & m_1 
\end{pmatrix}
\begin{pmatrix} \tau^2_0 \\ \sigma^2_0 \end{pmatrix}  
\label{eqn:critpoint} 
\end{equation} 
where $m_k(\tau^2,\sigma^2)$, $k\in \{1,2\}$ is given by 
\[ m_k(\tau^2,\sigma^2) =
\frac{ \int  \lambda^k  (\lambda \tau^2 +\sigma^2)^{-2} \, dF_0(\lambda) }
     { \int    (\lambda \tau^2 + \sigma^2 )^{-2} \, dF_0(\lambda) }.
\]
For each $(\tau^2,\sigma^2)\in [0,1)\times (0,\infty)$, 
$m_1$ and $m_2$ are the first and second moments of $\lambda$ 
under a probability measure having  density with respect to 
$F_0$ proportional to $(\lambda\tau^2 + \sigma^2)^{-2}$. 
If $F_0$ is not degenerate, then 
$m_2>m_1^2$, and so the determinant of the  matrix 
 in (\ref{eqn:critpoint})  is non-zero. 
Therefore, 
the matrix is invertible and  
$(\tau^2_0,\sigma_0^2)$  is the only solution to (\ref{eqn:critpoint}). 
The matrix of second derivatives of $Q_0$ is given by 
\[
  \frac{\partial^2 Q_0}{ \partial \theta \partial \theta^\top } 
 = \int \begin{pmatrix} \lambda^2 & \lambda \\ \lambda & 1 \end{pmatrix} 
    \frac{2(\lambda \tau^2_0 + \sigma^2_0) - (\lambda \tau^2 + \sigma^2)}
  { \lambda \tau^2+ \sigma^2 }  \times (\lambda \tau^2 + \sigma^2)^{-2} 
  \, dF_0(\lambda) . 
\]
At the critical point $(\tau^2_0,\sigma^2_0)$ this simplifies to 
 the expectation of the matrix in the integrand with respect to 
the probability measure with density 
proportional to $(\lambda\tau^2 + \sigma^2)^{-2}$ with respect to 
$F_0$. Again, if $F_0$ is not degenerate then the expectation of 
this matrix, and hence the Hessian of $Q_0$, is strictly positive definite.  
The critical point is a local minimum, and since it is the only critical 
point of the continuous function $Q_0$, it is the unique minimizer. 
This completes the proof of Lemma 4. 

To see how this applies to the properties of the empirical Bayes 
estimates 
$(\tilde\tau^2,\tilde \sigma^2)$ of $(\tau^2,\sigma^2)$ based on the 
marginal model (\ref{eqn:mmod}), 
let $\bl U$ be the $(p-1)\times (p-1)$ matrix of 
left singular vectors of $\bl X_2$, and let
 $n\bs\Lambda_2$ be the diagonal matrix of the squared singular values. 
Then $\bl U_2^\top \bl z_2/\sqrt{n} \sim N_{p-1}(\bl 0 , \bs\Lambda_2 
\tau^2 +\bl I \sigma_\infty^2 )$, and so the properties of the MLE 
of $(\tau^2,\sigma^2_\infty)$ based on $\bl z_2$ will  be the 
same as those of $(\tau^2,\sigma^2)$ in Lemma 4 if $\bs\Lambda_2$
satisfies the assumption of the Lemma. To see that it does, 
recall that the assumption of Lemma 2 was that 
the empirical distribution of the eigenvalues of $\bl X^\top \bl X/n$ 
is uniformly bounded and converges
weakly to a non-degenerate distribution with finite support. 
For a given $n$, let $\gamma_1,\ldots, \gamma_p$ be the eigenvalues of 
$\bl X^\top\bl X /n$. Since $\bl X_2^\top\bl X_2/n$ is a compression 
of $\bl X^\top \bl X/n$,
by the Cauchy interlacing theorem we have 
$\gamma_1\leq  \lambda_1 \leq \gamma_2 \leq  \cdots  \leq  \gamma_{p-1} \leq
 \lambda_{p-1} \leq \gamma_p$. 
Therefore, if the values of $\{\gamma_1,\ldots, \gamma_p\}$ 
are bounded uniformly in $p$ and 
have an empirical distribution that converges to a nondegenerate limit, 
then the same properties hold for the 
values of $\{\lambda_1,\ldots, \lambda_{p-1}\}$.

%% Note $\tau^2=0$ can be allowed if space for $\lambda$ is compact. 

\bibliographystyle{chicago}
\bibliography{\string~/Dropbox/Common/refs}

\begin{thebibliography}{}

\bibitem[\protect\citeauthoryear{Andrews}{Andrews}{1992}]{andrews_1992}
Andrews, D. W.~K. (1992).
\newblock Generic uniform convergence.
\newblock {\em Econometric Theory\/}~{\em 8\/}(2), 241--257.

\bibitem[\protect\citeauthoryear{B\"uhlmann}{B\"uhlmann}{2013}]{buhlmann_2013}
B\"uhlmann, P. (2013).
\newblock Statistical significance in high-dimensional linear models.
\newblock {\em Bernoulli\/}~{\em 19\/}(4), 1212--1242.

\bibitem[\protect\citeauthoryear{Conlon, Liu, Lieb, and Liu}{Conlon
  et~al.}{2003}]{conlon_etal_2003}
Conlon, E.~M., X.~S. Liu, J.~D. Lieb, and J.~S. Liu (2003).
\newblock Integrating regulatory motif discovery and genome-wide expression
  analysis.
\newblock {\em Proceedings of the National Academy of Sciences\/}~{\em
  100\/}(6), 3339--3344.

\bibitem[\protect\citeauthoryear{Efron, Hastie, Johnstone, and
  Tibshirani}{Efron et~al.}{2004}]{efron_etal_2004}
Efron, B., T.~Hastie, I.~Johnstone, and R.~Tibshirani (2004).
\newblock Least angle regression.
\newblock {\em Ann. Statist.\/}~{\em 32\/}(2), 407--499.
\newblock With discussion, and a rejoinder by the authors.

\bibitem[\protect\citeauthoryear{Kabaila and Tissera}{Kabaila and
  Tissera}{2014}]{kabaila_dilshani_2014}
Kabaila, P. and D.~Tissera (2014).
\newblock Confidence intervals in regression that utilize uncertain prior
  information about a vector parameter.
\newblock {\em Australian \& New Zealand Journal of Statistics\/}~{\em
  56\/}(4), 371--383.

\bibitem[\protect\citeauthoryear{Lee, Sun, Sun, and Taylor}{Lee
  et~al.}{2016}]{lee_sun_sun_taylor_2016}
Lee, J.~D., D.~L. Sun, Y.~Sun, and J.~E. Taylor (2016).
\newblock Exact post-selection inference, with application to the lasso.
\newblock {\em Ann. Statist.\/}~{\em 44\/}(3), 907--927.

\bibitem[\protect\citeauthoryear{Meinshausen, Meier, and
  B\"uhlmann}{Meinshausen et~al.}{2009}]{meinshausen_meier_buhlman_2009}
Meinshausen, N., L.~Meier, and P.~B\"uhlmann (2009).
\newblock {$p$}-values for high-dimensional regression.
\newblock {\em J. Amer. Statist. Assoc.\/}~{\em 104\/}(488), 1671--1681.

\bibitem[\protect\citeauthoryear{Newey and McFadden}{Newey and
  McFadden}{1994}]{newey_mcfadden_1994}
Newey, W.~K. and D.~McFadden (1994).
\newblock Large sample estimation and hypothesis testing.
\newblock In {\em Handbook of econometrics, {V}ol.\ {IV}}, Volume~2 of {\em
  Handbooks in Econom.}, pp.\  2111--2245. North-Holland, Amsterdam.

\bibitem[\protect\citeauthoryear{O'Gorman}{O'Gorman}{2001}]{ogorman_2001}
O'Gorman, T.~W. (2001).
\newblock Using adaptive weighted least squares to reduce the lengths of
  confidence intervals.
\newblock {\em Canad. J. Statist.\/}~{\em 29\/}(3), 459--471.

\bibitem[\protect\citeauthoryear{Pratt}{Pratt}{1963}]{pratt_1963}
Pratt, J.~W. (1963).
\newblock Shorter confidence intervals for the mean of a normal distribution
  with known variance.
\newblock {\em The Annals of Mathematical Statistics\/}~{\em 34\/}(2),
  574--586.

\bibitem[\protect\citeauthoryear{Ranga~Rao}{Ranga~Rao}{1962}]{rangarao_1962}
Ranga~Rao, R. (1962).
\newblock Relations between weak and uniform convergence of measures with
  applications.
\newblock {\em Ann. Math. Statist.\/}~{\em 33}, 659--680.

\bibitem[\protect\citeauthoryear{van~de Geer, B\"uhlmann, Ritov, and
  Dezeure}{van~de Geer et~al.}{2014}]{degeer_etal_2014}
van~de Geer, S., P.~B\"uhlmann, Y.~Ritov, and R.~Dezeure (2014).
\newblock On asymptotically optimal confidence regions and tests for
  high-dimensional models.
\newblock {\em Ann. Statist.\/}~{\em 42\/}(3), 1166--1202.

\bibitem[\protect\citeauthoryear{Yu and Hoff}{Yu and Hoff}{2016}]{yu_hoff_2016}
Yu, C. and P.~Hoff (2016).
\newblock Adaptive multigroup confidence intervals with constant coverage.

\bibitem[\protect\citeauthoryear{Zhang and Zhang}{Zhang and
  Zhang}{2014}]{zhang_zhang_2014}
Zhang, C.-H. and S.~S. Zhang (2014).
\newblock Confidence intervals for low dimensional parameters in high
  dimensional linear models.
\newblock {\em J. R. Stat. Soc. Ser. B. Stat. Methodol.\/}~{\em 76\/}(1),
  217--242.

\end{thebibliography}

\end{document}